\def\ba{\begin{array}}
	\def\ea{\end{array}}
\def\be{\begin{equation}}
\def\ee{\end{equation}}
\def\bg{\begin{aligned}}
	\def\eg{\end{aligned}}
\documentclass[pra,aps,showpacs,twocolumn,twoside,superscriptaddress]{revtex4}

\usepackage{amsmath,amsfonts,amssymb,caption,color,epsfig,graphics,graphicx,hyperref,latexsym,mathrsfs,revsymb,theorem,url,verbatim,epstopdf}

\usepackage{amssymb}
\usepackage{color}
\usepackage{amsmath,epic,curves,amscd}
\usepackage[english]{babel}
\usepackage{graphicx}
\usepackage{comment}
\usepackage{appendix}
\usepackage{mathdots}

\newtheorem{definition}{Definition}
\newtheorem{proposition}[definition]{Proposition}
\newtheorem{lemma}[definition]{Lemma}

\newtheorem{theorem}[definition]{Theorem}
\newtheorem{corollary}[definition]{Corollary}
\newtheorem{conjecture}[definition]{Conjecture}

\newtheorem{remark}[definition]{Remark}
\newtheorem{example}[definition]{Example}
\newtheorem{question}[definition]{Question}

\def\squareforqed{\hbox{\rlap{$\sqcap$}$\sqcup$}}
\def\qed{\ifmmode\squareforqed\else{\unskip\nobreak\hfil
		\penalty50\hskip1em\null\nobreak\hfil\squareforqed
		\parfillskip=0pt\finalhyphendemerits=0\endgraf}\fi}
\def\endenv{\ifmmode\;\else{\unskip\nobreak\hfil
		\penalty50\hskip1em\null\nobreak\hfil\;
		\parfillskip=0pt\finalhyphendemerits=0\endgraf}\fi}
\newenvironment{proof}{\noindent \textbf{{Proof.~} }}{\qed}
\def\Dbar{\leavevmode\lower.6ex\hbox to 0pt
	{\hskip-.23ex\accent"16\hss}D}
\makeatletter
\def\url@leostyle{%
	\@ifundefined{selectfont}{\def\UrlFont{\sf}}{\def\UrlFont{\small\ttfamily}}}
\makeatother
\def\bcj{\begin{conjecture}}
	\def\ecj{\end{conjecture}}
\def\bcr{\begin{corollary}}
	\def\ecr{\end{corollary}}
\def\bd{\begin{definition}}
	\def\ed{\end{definition}}
\def\bea{\begin{eqnarray}}
\def\eea{\end{eqnarray}}
\def\bem{\begin{enumerate}}
	\def\eem{\end{enumerate}}
\def\bex{\begin{example}}
	\def\eex{\end{example}}
\def\bim{\begin{itemize}}
	\def\eim{\end{itemize}}
\def\bl{\begin{lemma}}
	\def\el{\end{lemma}}
\def\bma{\begin{bmatrix}}
	\def\ema{\end{bmatrix}}
\def\bpf{\begin{proof}}
	\def\epf{\end{proof}}
\def\bpp{\begin{proposition}}
	\def\epp{\end{proposition}}
\def\bqu{\begin{question}}
	\def\equ{\end{question}}
\def\br{\begin{remark}}
	\def\er{\end{remark}}
\def\bt{\begin{theorem}}
	\def\et{\end{theorem}}

\def\btb{\begin{tabular}}
	\def\etb{\end{tabular}}

\newcommand{\nc}{\newcommand}

\def\a{\alpha}
\def\b{\beta}
\def\g{\gamma}
\def\d{\delta}

\def\t{\theta}

\def\m{\mu}
\def\n{\nu}

\def\r{\rho}
\def\s{\sigma}

\def\ps{\psi}

\def\G{\Gamma}

\def\Ph{\Phi}
\def\Ps{\Psi}

\nc{\bbA}{\mathbb{A}} \nc{\bbB}{\mathbb{B}} \nc{\bbC}{\mathbb{C}}
\nc{\bbD}{\mathbb{D}} \nc{\bbE}{\mathbb{E}} \nc{\bbF}{\mathbb{F}}
\nc{\bbG}{\mathbb{G}} \nc{\bbH}{\mathbb{H}} \nc{\bbI}{\mathbb{I}}
\nc{\bbJ}{\mathbb{J}} \nc{\bbK}{\mathbb{K}} \nc{\bbL}{\mathbb{L}}
\nc{\bbM}{\mathbb{M}} \nc{\bbN}{\mathbb{N}} \nc{\bbO}{\mathbb{O}}
\nc{\bbP}{\mathbb{P}} \nc{\bbQ}{\mathbb{Q}} \nc{\bbR}{\mathbb{R}}
\nc{\bbS}{\mathbb{S}} \nc{\bbT}{\mathbb{T}} \nc{\bbU}{\mathbb{U}}
\nc{\bbV}{\mathbb{V}} \nc{\bbW}{\mathbb{W}} \nc{\bbX}{\mathbb{X}}
\nc{\bbZ}{\mathbb{Z}}

\nc{\bA}{{\bf A}} \nc{\bB}{{\bf B}} \nc{\bC}{{\bf C}}
\nc{\bD}{{\bf D}} \nc{\bE}{{\bf E}} \nc{\bF}{{\bf F}}
\nc{\bG}{{\bf G}} \nc{\bH}{{\bf H}} \nc{\bI}{{\bf I}}
\nc{\bJ}{{\bf J}} \nc{\bK}{{\bf K}} \nc{\bL}{{\bf L}}
\nc{\bM}{{\bf M}} \nc{\bN}{{\bf N}} \nc{\bO}{{\bf O}}
\nc{\bP}{{\bf P}} \nc{\bQ}{{\bf Q}} \nc{\bR}{{\bf R}}
\nc{\bS}{{\bf S}} \nc{\bT}{{\bf T}} \nc{\bU}{{\bf U}}
\nc{\bV}{{\bf V}} \nc{\bW}{{\bf W}} \nc{\bX}{{\bf X}}
\nc{\bZ}{{\bf Z}}

\nc{\cA}{{\cal A}} \nc{\cB}{{\cal B}} \nc{\cC}{{\cal C}}
\nc{\cD}{{\cal D}} \nc{\cE}{{\cal E}} \nc{\cF}{{\cal F}}
\nc{\cG}{{\cal G}} \nc{\cH}{{\cal H}} \nc{\cI}{{\cal I}}
\nc{\cJ}{{\cal J}} \nc{\cK}{{\cal K}} \nc{\cL}{{\cal L}}
\nc{\cM}{{\cal M}} \nc{\cN}{{\cal N}} \nc{\cO}{{\cal O}}
\nc{\cP}{{\cal P}} \nc{\cQ}{{\cal Q}} \nc{\cR}{{\cal R}}
\nc{\cS}{{\cal S}} \nc{\cT}{{\cal T}} \nc{\cU}{{\cal U}}
\nc{\cV}{{\cal V}} \nc{\cW}{{\cal W}} \nc{\cX}{{\cal X}}
\nc{\cY}{{\cal Y}}
\nc{\cZ}{{\cal Z}}

\nc{\hA}{{\hat{A}}} \nc{\hB}{{\hat{B}}} \nc{\hC}{{\hat{C}}}
\nc{\hD}{{\hat{D}}} \nc{\hE}{{\hat{E}}} \nc{\hF}{{\hat{F}}}
\nc{\hG}{{\hat{G}}} \nc{\hH}{{\hat{H}}} \nc{\hI}{{\hat{I}}}
\nc{\hJ}{{\hat{J}}} \nc{\hK}{{\hat{K}}} \nc{\hL}{{\hat{L}}}
\nc{\hM}{{\hat{M}}} \nc{\hN}{{\hat{N}}} \nc{\hO}{{\hat{O}}}
\nc{\hP}{{\hat{P}}} \nc{\hR}{{\hat{R}}} \nc{\hS}{{\hat{S}}}
\nc{\hT}{{\hat{T}}} \nc{\hU}{{\hat{U}}} \nc{\hV}{{\hat{V}}}
\nc{\hW}{{\hat{W}}} \nc{\hX}{{\hat{X}}} \nc{\hZ}{{\hat{Z}}}

\nc{\hn}{{\hat{n}}}

\def\dim{\mathop{\rm Dim}}

\def\lin{\mathop{\rm span}}

\def\max{\mathop{\rm max}}

\def\rank{\mathop{\rm rank}}

\def\tr{\mathop{\rm Tr}}

\def\dg{\dagger}

\def\op{\oplus}

\def\ra{\rightarrow}

\newcommand{\bra}[1]{\langle#1|}
\newcommand{\ket}[1]{|#1\rangle}
\newcommand{\proj}[1]{| #1\rangle\!\langle #1 |}
\newcommand{\ketbra}[2]{|#1\rangle\!\langle#2|}

\newcommand{\abs}[1]{|#1|}

\def\Dbar{\leavevmode\lower.6ex\hbox to 0pt
	{\hskip-.23ex\accent"16\hss}D}

\begin{document}
	\title{Additivity of entanglement of formation via entanglement-breaking space}
	
	\date{\today}
	
	
	\author{Li-Jun Zhao}\email[]{zhaolijun@buaa.edu.cn}
	\affiliation{School of Mathematics and Systems Science, Beihang University, Beijing 100191, China}
	
	\author{Lin Chen}\email[]{linchen@buaa.edu.cn (corresponding author)}
	\affiliation{School of Mathematics and Systems Science, Beihang University, Beijing 100191, China}
	\affiliation{International Research Institute for Multidisciplinary Science, Beihang University, Beijing 100191, China}

\begin{abstract}
	We study the entanglement-breaking (EB) space, such that the entanglement of formation (EOF) of a bipartite quantum state is additive when its range is an EB subspace.
	We systematically construct the EB spaces in the Hilbert space $\bbC^m\otimes\bbC^3$, and the $2$-dimensional EB space in $\bbC^2\otimes\bbC^n$. We characterize the expression of two-qubit states of rank two with nonadditive EOF, if they exist. We further
	apply our results to construct EB spaces of an arbitrarily given dimensions. We show that the example in [PRL 89,027901(2002)] is a special case of our results. We further work out the entanglement cost of a qubit-qutrit state in terms of the two-atom system of the Tavis-Cummings model.
\end{abstract}
	\maketitle


\section{Introduction}
The entanglement  of formation (EOF) has been constructed to quantify the amount of  quantum communication for creating entangled state \cite{PhysRevA.54.3824}.
As a well-known entanglement measure, it has been widely useful in quantum information over the past decades \cite{PhysRevA.61.062102,PhysRevLett.85.2625, PhysRevLett.95.210501,PhysRevA.63.042306, PhysRevLett.121.190503}. The EOF is related to other entanglement measures such as the geometric measure of entanglement, relative entropy of entanglement, and robustness of entanglement \cite{Zhu2010Additivity}. The EOF has also been studied in
genuine quantum correlations of many-body systems, monogamy property even spin coherent states, Einstein-Podolsky-Rosen–like correlations and EOF of symmetric Gaussian states \cite{0034-4885-81-7-074002, 1751-8121-46-39-395302, PhysRevLett.91.107901}. Quantifying EOF is a key step in these problems and applications.

Hitherto, the explicit analytic formulae for EOF have been found for two-qubit system
\cite{wootters1998} and some special states such as the isotropic states and some $16\times16$ mixed states \cite{PhysRevLett.85.2625, D2018Exact,FEI2003333}.
Nevertheless, the EOF of the tensor product of two states is usually hard to compute, even if their EOF is known. The reason is that
the EOF is not additive \cite{Hastings2009Superadditivity}, though
such states are not constructive
yet. In this paper we shall investigate the additivity of EOF via the so-called entanglement breaking (EB) space.

Apart from the computing of EOF of tensor product of two states, the additivity of EOF of many states is also highly expected for the following reasons.
First, entanglement cost and distillable entanglement are two asymptotic limit of entanglement measurement with physical significance in practice \cite{RevModPhys.81.865}. The entanglement cost of preparing a state $\rho$ is equal to $E_c(\r)=\lim\limits_{n\rightarrow\infty}{1\over n}E_{f}(\rho^{\otimes n})$ where $E_{f}$ is the EOF \cite{0305-4470-34-35-314}. If $\r$ is additive then we can obtain $E_c(\r)=E_f(\r)$, and thus the problem of computing $E_c(\r)$ is kind of simplified.
There are also equivalent conditions on the additivity of EOF and minimum entropy \cite{PhysRevA.75.060304,PhysRevA.98.042338, PhysRevA.68.032317,Shor2004}.
Next, it is known that every entanglement measure can be constructed as a coherence measure, and vice versa \cite{PhysRevLett.115.020403,PhysRevA.94.022329, PhysRevLett.121.220401}. By measuring coherence correlation, we can construct the counterpart of entanglement measures, such as EOF, entanglement cost and so on.
So the additivity of coherence measures in the resource theory is intimately connected to that of EOF. We can link various quantum correlation measures in resource theory and better study the properties of quantum states \cite{PhysRevA.84.012313}.

\begin{figure}[htb]
	\includegraphics[scale=0.3,angle=0]{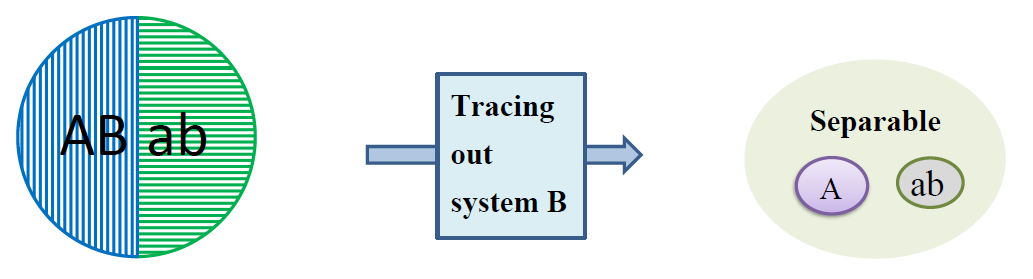}
	\caption{Four arbitrary quantum systems, denoted $A$, $B$, $a$ and $b$. Suppose tracing out system $B$ destroys the entanglement between $AB$ and $ab$. We shall refer to such a bipartite subspace in $\cH_{AB}$ as an entanglement-breaking (EB) space. The blue stripes represent the bipartite EB space and the green stripes represent any bipartite space $\mathcal{H}_{ab}$. The bipartite state whose range is an EB space has additive EOF.}
	\label{fig:binary}
\end{figure}

There have been studies highlighting the additivity of the EOF for any state supported on specific spaces \cite{PhysRevLett.89.027901}. The method relies on the assumption that tracing out one of the parties corresponds to an entanglement-breaking channel. The paper \cite{PhysRevLett.89.027901} has constructed two examples for illustrating their method. In this paper, we shall systematically develop their method by EB spaces. The latter is defined as follows. If the range of a bipartite state is an EB space, then the state has additive EOF. We refer readers to Figure \ref{fig:binary} and Definition \ref{df:ebc} for more details.

As a primary example, we show in Lemma \ref{le:ebspace} that any $1$-dimensional bipartite subspace is an EB space. Then we provide the fundamental facts on EB spaces and positive-partial-transpose (PPT) states in Lemma \ref{le:eof=ebspace} and \ref{le:rank2}. We point out that
the bipartite EB subspace in $\bbC^m\otimes\bbC^n$ has dimension at most $n$, and the upper bound is saturated in Lemma \ref{le:ebdim}. Then we construct the $2$-dimensional EB space in $\bbC^m\otimes\bbC^2$ in Lemma \ref{le:2dim=2qubit}, and extend it to high dimensions. In particular, we construct the EB spaces in $\bbC^m\otimes\bbC^3$, and the $2$-dimensional EB space in $\bbC^2\otimes\bbC^n$ in Theorem \ref{thm:mx3} and Corollary \ref{cr:2xm}, respectively. As a preliminary result, we shall investigate the $2$-dimensional EB space in $\bbC^m\otimes\bbC^3$ by Theorem \ref{thm:2dim=3x2}.
As an application of our result, we characterize in Lemma \ref{le:nonadditive} the expression of two-qubit states of rank two with non-additive EOF, if they exist. It sheds light on the open problem of constructing states with non-additive EOF.

To provide systematic methods of constructing EB spaces of high dimensions, we show that the tensor product of two EB spaces is also an EB space in Theorem \ref{thm:votimesw}. We extend it to the tensor product of arbitrarily many EB spaces, and the $B$-direct sum of EB spaces in Corollary \ref{cr:votimesw}. Since any subspace of an EB space is an EB space, we can construct EB spaces of any dimension.

We also consider a specific scenario, namely the Tavis-Cummings (TC) model where the qubit-qutrit state can be physically realized. We calculate analytically the EOF as a function of time for the qubit-qutrit mixed state \cite{PhysRevA.85.022320}. We further prove that the TC model belongs to a 2-dimensional EB space, so it satisfies additivity, and the entanglement cost of TC model can be calculated.

The rest of this paper is as follows. In
Sec. \ref{sec:pre} we introduce the preliminary knowledge we need in this paper. In Sec. \ref{sec:main} we introduce the main result. In Sec. \ref{sec:higher} we apply our results to construct more EB spaces of higher dimensions. In Sec. \ref{sec:moreapp} we introduce the physical application of our results. Finally we conclude in Sec. \ref{sec:con}.

\section{Preliminaries}
\label{sec:pre}
For two $n$-partite states $\r,\s$,
suppose there is an $n$-partite product matrix $U=\otimes^n_{j=1}U_j$ such that
$U^\dg\r U=\s$. We say that $\r$ is convertible to $\s$ under stochastic local operations and classical communications (SLOCC). If $U$ is unitary then we say that $\r$ and $\s$ are locally unitarily (LU) equivalent. The two kinds of equivalence are both realizable in experiments, and extensively useful in quantum information \cite{Baerdemacker2017The}. For example, the three-qubit pure states have been classified in terms both SLOCC and LU equivalence, and the former generates the known Greenberger-Horne-Zeilinger (GHZ) and W states \cite{dvc2000}. In the following, we review the space for studying the additivity of EOF.
\begin{definition}
	\label{df:ebc}	
	(i) We refer to $A$, $B$, $a$ and $b$ as four quantum systems, and a subspace $V\subseteq\mathcal{H}_{AB}$. Let $\ket{\Ps}_{Vab}\in V\otimes\cH_{ab}$ be an arbitrary bipartite pure state. Suppose that tracing out system $B$ destroys all entanglement between $AB$ with $ab$, i.e.,
	\begin{equation}\label{EOFth}
	\tr{_B}(\ket{\Psi}_{Vab}\bra{\Psi})=\sum_{n}q_{n}\ket{\mu_{n}}_{A}\bra{\mu_{n}}\otimes \ket{\nu_{n}}_{ab}\bra{\nu_{n}}.
	\end{equation}
	We shall refer to such $V$ as an entanglement-breaking (EB) space.
	
	(ii) Let $V=\lin\{\ket{a_1},...,\ket{a_m}\}$ and $V'=\lin\{\ket{b_1},...,\ket{b_n}\}$ be two subspaces in $\cH_{AB}$, and they respectively have dimension $m$ and $n$. We say that $V$ is EB-convertible to $V'$, if there exists a product matrix $W_A\otimes U_B$ with a matrix $W$ and a unitary matrix $U$, such that
	\begin{eqnarray}
	\label{eq:convertible}
	V'=
	\lin\{
	(W_A\otimes U_B)\ket{a_i},\quad i=1,...,m\}.
	\end{eqnarray}
	Further, we say that the two spaces $V$ and $V'$ are EB-equivalent if the matrix $W$ is invertible.
	\qed
\end{definition}
The definition implies that if $V$ is EB-convertible to $V'$ then $\ket{b_j}$ is the linear combination of
the states $(W_A\otimes U_B)\ket{a_i}$'s. So $\dim V\ge \dim V'$. As a result, two EB-equivalent spaces have the same dimension.

Part (i) of Definition \ref{df:ebc} was constructed in \cite{PhysRevLett.89.027901}. It implies the following fact.
\begin{lemma}
	\label{le:ebspace}
	(i) Any subspace of an EB space is still an EB space.
	
	(ii) If $V$ is an EB space, then so is $(W_A\otimes U_B)V$ for any matrix $W$ and any unitary matrix $U$. The converse also holds if $W$ is invertible.
	
	(iii) Any $1$-dimensional bipartite space is an EB space.
\end{lemma}
Since assertion (i) and (ii) follow from the definition of EB spaces, we only prove (iii). Suppose $V\subseteq\cH_{AB}$ is a $1$-dimensional subspace spanned by $\ket{\ps}$. So any state in $V\otimes\cH_{ab}$ can be written as $\ket{\ps}_{AB}\otimes\ket{\phi}_{ab}$.
By tracing out the system $B$, we obtain a product state between system $A$ and $ab$. Definition \ref{df:ebc}	implies that $V$ is an EB space.

We demonstrate the above definition and lemma by studying the relation of the two EB spaces $U$ and $V$ proposed in \cite{PhysRevLett.89.027901}.
They are respectively spanned by
\footnote{Eq. \eqref{eq:6*3PRL} is originally obtained by Eq. (28) of \cite{PhysRevLett.89.027901}, which has a misprint in the third state $\ket{2}_U$. That is, $\ket{2,5}$ was wrongly written as $\ket{0,5}$ in Eq. (28) of \cite{PhysRevLett.89.027901}. We have corrected the misprint in Eq. \eqref{eq:6*3PRL}.},
\begin{eqnarray}
\label{eq:6*3PRL}
	\begin{aligned}
\ket{0}_U&=\frac{1}{2}(\ket{1,2}+\ket{2,1}+\sqrt{2}\ket{0,3}),\\
\ket{1}_U&=\frac{1}{2}(\ket{2,0}+\ket{0,2}+\sqrt{2}\ket{1,4}),\\
\ket{2}_U&=\frac{1}{2}(\ket{0,1}+\ket{1,0}+\sqrt{2}\ket{2,5}),
\end{aligned}
\end{eqnarray}
and
\begin{eqnarray}
\label{2*3PRL}
\begin{aligned}
\ket{0}_{V}&=\frac{1}{\sqrt{3}}(\ket{0,2}-\sqrt{2}\ket{1,0}),\\
\ket{1}_{V}&=\frac{-1}{\sqrt{3}}(\ket{1,2}-\sqrt{2}\ket{0,1}).\\	
\end{aligned}		
\end{eqnarray}
Let the product matrix $W=(\ketbra{0}{1}-\ketbra{1}{0})\otimes(\ketbra{0}{3}+\ketbra{1}{4}+\proj{2}+\ketbra{4}{1}+\ketbra{3}{0}+\proj{5})$. One can verify that $W\ket{0}_U\propto\ket{0}_V$ and $W\ket{1}_U\propto\ket{1}_V$. Lemma \ref{le:ebspace} shows that the subspace spanned by $\ket{0}_U,\ket{1}_U$ is an EB space. Definition \ref{df:ebc} shows that
the subspace is EB convertible to the subspace spanned by $\ket{0}_V,\ket{1}_V$. It is also an EB space by Lemma \ref{le:ebspace}. We have shown the conclusion of \cite{PhysRevLett.89.027901} using Definition \ref{df:ebc} and Lemma \ref{le:ebspace}.

The importance of EB space is that, the bipartite state has additive EOF when its range is an EB space \cite{PhysRevLett.89.027901}. That is,

\begin{lemma}\label{le:eof=ebspace}
	Let $\rho_{V}\in\mathcal{B}(V)$ such that $V\subseteq\cH_A\otimes\cH_B$ is an EB subspace. Then the EOF of $\rho_{V}$ is additive, i.e.,
	\begin{eqnarray}
	E_{f}(\rho_{V}\otimes \sigma)=E_{f}(\rho_{V})+E_{f}(\sigma),
	\end{eqnarray}
where the bipartite state $\r_V\otimes\sigma\in\cB(\cH_{Aa}\otimes\cH_{Bb})$,
	and $\sigma\in\cB(\cH_{ab})$ is an arbitrary bipartite state.
	
	As a corollary, by setting $\r_V^{\otimes k}=\s$ for any integer $k\ge1$ we have $E_{f}(\rho_{V})=E_{c}(\rho_{V})$, i.e., the EOF and entanglement cost of $\r_V$ are the same.
\end{lemma}
The result shows that finding an EB space is a systematic method of constructing entangled states with additive EOF. Recall that Ref. \cite{PhysRevLett.89.027901}  has proposed only two examples, namely \eqref{eq:6*3PRL} and \eqref{2*3PRL}. Further, it is known that if the EOF of two states are additive, then so is the EOF of their tensor product  \cite{Chen2007Entanglement}. If we choose $\s$ with additive EOF in \eqref{le:eof=ebspace}, then we can obtain the bipartite state $\r_V\otimes\s$ with additive EOF, even if the range of $\s$ is not in an EB space. This is why we want to characterize EB spaces in this paper.

For introducing the techniques  in the upcoming sections, we refer to an $m\times n$ state $\r$ in the sense that $\rank\r_A=m$ and $\rank\r_B=n$. Further, the Peres-Horodecki criterion says that if a bipartite state is separable then it has positive partial transpose (PPT), and the converse holds for qubit-qutrit states \cite{peres1996,hhh96}. The following assertions have been widely used in literatures \cite{cd12,Chen2013The, Lin2013Separability}, and can be proven straightforwardly using the Peres-Horodecki criterion.
\begin{lemma}
	\label{le:rank2}
	(i) The bipartite state of rank at most three is separable if and only if it is PPT.

	(ii) The $m\times n$ state of rank $\max\{m,n\}$ is separable if and only if it is PPT. In this case the state is the convex sum of $\max\{m,n\}$ pure product states.
	
	(iii) If $\ket{\a}$ is a bipartite pure entangled state and $\b=\proj{\g}\otimes\d$ is a bipartite mixed state, then the bipartite state $\proj{\a}+\proj{\g}\otimes\d$ is entangled. 	
	
	(iv) The $m\times n$ state of rank $r$ is entangled if $\max\{m,n\}>r$.
\end{lemma}

Combining with the definition of EB spaces, Lemma \ref{le:rank2} (iv) shows that
\begin{lemma}
\label{le:ebdim}	
The bipartite EB subspace in $\bbC^m\otimes\bbC^n$ has dimension at most $n$. The upper bound is saturated when the subspace is spanned by $\ket{a_1,1},...,\ket{a_n,n}$.
\end{lemma}

Equivalently, any $(n+1)$-dimensional subspace in $\bbC^m\otimes\bbC^n$ is not an EB space. In fact, if the subspace is spanned by $\ket{a_1},..,\ket{a_{n+1}}$ and $\ket{\ps}={1\over\sqrt{n+1}}\sum_j\ket{a_j,j}\in\cH_{AB:ab}$, then $\tr_B\proj{\ps}$ is a bipartite entangled state of system $A$ and $ab$.

\section{Main results}
\label{sec:main}

In Lemma \ref{le:ebspace} we have shown that any $1$-dimensional bipartite subspace is an EB space. However characterizing $n$-dimensional EB spaces with $n\ge2$ is an open problem. So far only a few examples have been constructed \cite{PhysRevLett.89.027901}. In this section we construct the EB space in $\bbC^m\otimes\bbC^3$, and the $2$-dimensional EB space in $\bbC^2\otimes\bbC^n$.
We begin by studying the $2$-dimensional EB space in $\bbC^m\otimes\bbC^2$, and then extend it to high dimensions.

\begin{lemma}
	\label{le:2dim=2qubit}
	For any positive integer $m$, the 2-dimensional subspace in $\bbC^m\otimes\bbC^2$ is an EB space if and only if it is spanned by
	$
	\{\ket{x,0}, \ket{y,1}\}$ up to LU equivalence.
\end{lemma}
\begin{proof}
	Suppose $V\subseteq \cH_A\otimes\cH_B =\bbC^m\otimes\bbC^2$ is a $2$-dimensional subspace. To prove the ``if" part, we assume that $V$ is spanned by $
	\{\ket{x,a_0}, \ket{y,a_1}\}$ where $\{\ket{a_0},\ket{a_1}\}$ is an orthonormal basis in $\bbC^2$. So any state in $V\otimes\cH_{ab}$ can be written as $\sqrt{p}\ket{x,a_0}_{AB}\otimes\ket{\phi_1}_{ab}+\sqrt{1-p}\ket{y,a_1}_{AB}\otimes\ket{\phi_2}_{ab}$.
	By tracing out the system $B$, we obtain the resulting the state is a bipartite state $p\proj{x}\otimes\proj{\phi_1}+(1-p)\proj{y}\otimes\proj{\phi_2}$. So it is separable, and Definition \ref{df:ebc}	implies that $V$ is an EB space.
	
	Next we prove the ``only if" part of assertion. Assume that $V$ is not spanned by $
	\{\ket{x,0}, \ket{y,1}\}$ up to LU equivalence, and we will prove that $V$ is not an EB space. We prove it in two cases, namely (i) $V$ has a product state, and (ii) $V$ does not have a product state.
	
	(i) Since $V\subseteq\bbC^m\otimes\bbC^2$, we may assume that $V$ is spanned by $\{\ket{\alpha_{1}, \alpha_{2}}, \ket{a,0}+\ket{c,1}\}$.  Using the LU equivalence, we can rewrite
	\begin{eqnarray}
	\begin{aligned}
	U\ket{\alpha_{1}, \alpha_{2}}=&\ket{0, 0},\\
	U(\ket{a,b}+\ket{c,d})=&\ket{s,0}+\ket{t,1},
	\end{aligned}
	\end{eqnarray} and $U$ is a product unitary operator. The assumption implies that
	$\ket{s}$ is not proportional to $\ket{0}$. We consider the state $
	\ket{\Psi}_{ABab}=\ket{0,0}_{AB}\ket{0}_{ab}+(\ket{s,0}+\ket{t,1})_{AB}\ket{1}_{ab}.
	$
	Since $\ket{s}$ is not proportional to $\ket{0}$, the reduced state $	
	{\tr}_{B}\proj{\Psi}_{ABab}$ is entangled in terms of Lemma \ref{le:rank2} (iii). So $V$ is not an EB space.
	
	(ii) Suppose $V$ does not have a product state. Since any two-qubit space has a product vector, we obtain $m\ge3$. Using Lemma \ref{le:ebspace}, the two basis vectors of $V$ can be written as the two non-normalized entangled states $\ket{0,0}+\ket{1,1}$ and $\ket{2,0}+\ket{a,1}$. Consider the following state $
	\ket{\Phi}_{Vab}=
	(\ket{0,0}+\ket{1,1})_{AB}\ket{0}_{ab}
	+
	(\ket{2,0}+\ket{a,1})_{AB}\ket{1}_{ab},
	$
	One can show that there is at most one product vector in the range of ${\tr}_{B}\proj{\Psi}_{Vab}$. So the latter is an entangled state, and $V$ is not an EB space. This completes the proof.
\end{proof}

The physical implication of this lemma is to understand the bipartite state $\r$ with non-additive EOF, since such a state is known to exist and not constructed yet \cite{Hastings2009Superadditivity}. The lemma shows that if $\r$ is a two-qubit state of rank two, then its range is not spanned by $\ket{x,0},\ket{y,1}$ up to EB equivalence. Since any $2$-dimensional two-qubit subspace has a product state, we obtain that
\begin{lemma}
	\label{le:nonadditive}
	If a two-qubit state of rank two has non-additive EOF, then it is the convex sum of ${1\over\sqrt2}(\sqrt{p}\ket{a,0}+\sqrt{1-p}\ket{b,1})$ and $\ket{a,1}$, where $\ket{a},\ket{b}$ are linearly independent qubit states.	
\end{lemma}

Based on Lemma \ref{le:2dim=2qubit}, we characterize the $2$-dimensional EB subspace in $\bbC^m\otimes\bbC^3$.

\begin{theorem}
	\label{thm:2dim=3x2}
	The $2$-dimensional subspace in $\bbC^m\otimes\bbC^3$ is an EB subspace if and only if one of the following three spaces is EB convertible to the subspace. Up to normalization factors, they are respectively spanned by
	$
	\{\ket{0,0}, \ket{1,1}+\ket{2,2}\}$, $\{\ket{0,0}+\ket{1,1},\ket{1,1}+\ket{2,2}\}$, and
	\begin{eqnarray}
	\label{eq:mx3,family3}
	&&
	\{
	\ket{0,0}+\ket{1,1},
	\quad
	\ket{0,1}
	+
	\ket{1}(d\ket{0}+fe^{i\t}\ket{1}
	+g\ket{2})
	\},
	\notag\\
	\end{eqnarray}
	where $g>0$ and $d,f,\t\ge0$ satisfy
	\begin{eqnarray}
	\label{eq:dfgt}
	&&
	-1-f^2+g^2-d^2(-2+d^2+f^2+g^2)
	\notag\\&&
	+2df^2\cos2\t\ge0.	
	\end{eqnarray}
	\qed
\end{theorem}
We refer readers to its proof in Appendix \ref{app:2dim=3x2}. We demonstrate Theorem \ref{thm:2dim=3x2} by showing that the EB subspace $V\subset\cH_{AB}$ in \eqref{2*3PRL} is spanned by \eqref{eq:mx3,family3} up to EB equivalence. It suffices to consider states $\ket{\Psi_1}_{Vab}\in V\otimes\cH_{ab}$, which is a bipartite entangled state of system $AB$ and $ab$. Up to an invertible transformation on system $ab$, the state can be written as
\begin{eqnarray}
\ket{\Psi_1}_{Vab}
=&
\frac{1}{\sqrt{3}}(\ket{0,2}-\sqrt{2}\ket{1,0})\ket{0}
\notag\\+&
\frac{1}{\sqrt{3}}(\sqrt{2}\ket{0,1}-\ket{1,2})\ket{1}.
\end{eqnarray}
Let $\Pi=\frac{\sqrt{3}}{\sqrt{2}}\bma
1 & 0 \\
0 & -\sqrt{2}
\ema \otimes
\bma
0&0&1\\
1&0&0\\
0&1&0\\
\ema
\otimes
\bma
0 & \sqrt{2} \\
1 & 0
\ema $ be a product operator on the system $A,B$ and $ab$. We have
\begin{eqnarray}
\begin{aligned}
\Pi\ket{\Psi_1}_{B:A:ab}=&\ket{0}\ket{0,0}
+
\ket{1}
(\ket{0,1}+\ket{1,0})\\&
+
(d\ket{0}+fe^{i\t}\ket{1}
+g\ket{2})\ket{1,1}.
\end{aligned}
\end{eqnarray}
when $d=fe^{i \theta}=0$, and $g=2$.
So the range of $A,B$ reduced state of $\Pi\ket{\Psi_1}_{B:A:ab}$ is a special case of \eqref{eq:mx3,family3}.

Now we are in a position to present the main result of this section. That is, we characterize any EB subspace in $\bbC^m\otimes\bbC^3$.
\begin{theorem}
	\label{thm:mx3}
	The EB subspace $V\subset\bbC^m\otimes\bbC^3$ has dimension one, two or three.
	
	(i) If $\dim V=1$, then $V$ is an arbitrary one-dimensional subspace in $\bbC^m\otimes\bbC^3$.
	
	(ii) If $\dim V=2$ then $V$ is spanned by $(W\otimes X)\cS$, where $W$ is an order-$3$ matrix, $X$ is a unitary matrix, and $\cS$ is one of the three bases in Theorem \ref{thm:2dim=3x2}.
	
	(iii) If $\dim V=3$ then $V$ is spanned by $\ket{x,1},\ket{y,2}$ and $\ket{z,3}$ up to LU equivalence.
\end{theorem}
\begin{proof}
	Lemma \ref{le:ebdim} shows that $\dim V\le 3$. Using Theorem \ref{thm:2dim=3x2}, we only need prove assertion (iii). Let $\ket{\Ps_2}_{Vab}\in\cH_{A:B:ab}$, such that $\r_{A:ab}=\tr_B\proj{\Ps_2}$ is separable and $\rank\r_{ab}=3$. Since $V\subset\bbC^m\otimes\bbC^3$, Lemma \ref{le:rank2} (iv) shows that $\rank\r_B=3$. Lemma \ref{le:rank2} (ii) shows that $\r_{A:ab}$ is the convex sum of three pure product states. Let $\r_{A:ab}=\sum^3_{j=1}p_j\proj{a_j,b_j}$.
	We obtain its purification, namely $\ket{\Ps_2}_{Vab}=\sum^3_{j=1}p_j\ket{a_j,j,b_j}$ up to LU equivalence.  So $V$ is spanned by $\ket{a_j,j}$'s, and assertion (iii) holds. 	
\end{proof}

So far we have investigated some EB spaces whose system $B$ has small dimension $2,3$. In the following we investigate an EB space whose $B$ has arbitrarily large dimension.

\begin{corollary}
	\label{cr:2xm}	
	The $2$-dimensional EB subspace $V$ in $\bbC^2\otimes\bbC^n$ is LU equivalent to those in $\bbC^2\otimes\bbC^k$ with $k=2,3$ or $4$
	\footnote{We say that two bipartite spaces $S$ and $T$ of the same dimension are LU-equivalent when there is a constant product unitary operation $U$ such that $U\ket{a}\in T$ for any $\ket{a}\in S$.}. The space $V$ with $k=2,3$ has been characterized in Theorem \ref{thm:mx3}. Up to normalization factors and EB equivalence, $V$ with $k=4$ is spanned by
	\begin{eqnarray}
	\label{eq:2xn}
	\ket{0,0}+	
	b_0
	(a_0\ket{0}+a_1\ket{1})
	\ket{2}
	+
	d_0
	(c_0\ket{0}+c_1\ket{1})
	\ket{3},
	\notag\\
	\ket{1,1}+	
	b_1
	(a_0\ket{0}+a_1\ket{1})
	\ket{2}
	+
	d_1
	(c_0\ket{0}+c_1\ket{1})
	\ket{3},
	\end{eqnarray}
	for any complex numbers $a_j,b_j,c_j,d_j$ for $j=0,1$ and
	\begin{eqnarray}
	\label{eq:det}
	\det
	\bma
	a_0b_1 & a_1b_0\\	
	c_0d_1 & c_1d_0\\	
	\ema\ne0.
	\end{eqnarray}
\end{corollary}
\begin{proof}
	It suffices to prove the assertion for $k=4$. Let $\ket{\Ps}\in\cH_{A:ab:B}$, and its reduced density operator $\r_{A:ab}$ is a two-qubit separable state of rank four. It is known that such a state is the convex sum of four pure product states \cite{wootters1998}. Up to EB equivalence, we may assume that
	\begin{eqnarray}
	\begin{aligned}
	\r_{A:ab}=&
	\proj{0,0}+\proj{1,1}\\&+
	(a_0\ket{0}+a_1\ket{1})	
	(\bra{0}+a_1^*\bra{1})
	\\&\otimes
	(b_0\ket{0}+b_1\ket{1})	
	(b_0^*\bra{0}+b_1^*\bra{1})	
	\\&+
	(c_0\ket{0}+c_1\ket{1})	
	(c_0^*\bra{0}+c_1^*\bra{1})
	\\&\otimes
	(d_0\ket{0}+d_1\ket{1})	
	(d_0^*\bra{0}+d_1^*\bra{1}).
	\end{aligned}	
	\end{eqnarray}
	Since $\ket{\Ps}\in\cH_{A:ab:B}$ is the purification of $\r_{A:ab}$, up to a local unitary gate on system $B$ we have
	\begin{eqnarray}
	\begin{aligned}
	\ket{\Ps}=
	&
	\ket{0,0,0}+
	\ket{1,1,1}	
	\\&+
	(a_0\ket{0}+a_1\ket{1})
	(b_0\ket{0}+b_1\ket{1})
	\ket{2}
	\\&+
	(c_0\ket{0}+c_1\ket{1})
	(d_0\ket{0}+d_1\ket{1})
	\ket{3}.
	\end{aligned}
	\end{eqnarray}
	So the range of $\tr_{ab}\proj{\Ps}$ is spanned by the basis in \eqref{eq:2xn}. Since $\r_{A:ab}$ has rank four, Eq. \eqref{eq:det} is satisfied. The range is exactly the space $V$, and the assertion holds.
\end{proof}

To conclude this section, we construct a $3$-dimensional EB space in $\bbC^m\otimes\bbC^n$ of any integer $m\ge3$ and $n=3N\ge6$. The space is spanned by the unnormalized basis states
\begin{eqnarray}
\begin{aligned}
&&\ket{1,2}+\ket{2,1}+x_0\ket{0,3}+\sum_{i=1}^{N-2}\ket{a_{i},3(i+1)},\\
&&\ket{2,0}+\ket{0,2}+x_1\ket{1,4}+\sum_{i=1}^{N-2}\ket{b_{i},3(i+1)+1},\\
&&\ket{0,1}+\ket{1,0}+x_2\ket{2,5}+\sum_{i=1}^{N-2}\ket{c_{i},3(i+1)+2}.
\notag\\
\end{aligned}
\end{eqnarray}
where $\abs{x_j}\ge\sqrt{2}$, and $\ket{a_i},\ket{b_i},\ket{c_i}$ are arbitrarily (unnormalized) states in $\bbC^m$. One can verify that the space is an EB space. Note that if $x_j=2$ and $N=2$, then the space reduces to the $3$-dimensional EB space $U$ spanned by \eqref{eq:6*3PRL}. In this sense, we have extended the finding in \cite{PhysRev.170.379, PhysRev.188.692}.

\section{Constructing EB spaces of any dimensions}
\label{sec:higher}

We have constructed a few EB spaces of small dimensions in the last section. A nature question is to construct EB space of any dimension. The first answer to this question is the space saturating the upper bound in Lemma \ref{le:ebdim}. The second answer is the EB space of dimension three constructed at the end of last section. Nevertheless, both of them are examples and do not provide systematic methods of constructing EB spaces. To handle the problem, in this section we investigate the tensor product of two EB spaces, and show that it is also an EB space in Theorem \ref{thm:votimesw}. We extend it to the tensor product of arbitrarily many EB spaces, as well as the $B$-direct sum of EB spaces in Corollary \ref{cr:votimesw}. So we can construct EB spaces of arbitrarily large dimension. Since any subspace of an EB space is still an EB space, we can thus construct EB spaces of any dimension. First of all we present the main result of this section.

\begin{theorem}
	\label{thm:votimesw}
	Suppose $V\subset\cH_{A_1B_1}$ and $W\subset\cH_{A_2B_2}$ are two EB subspaces. Then so is the bipartite subspace $V\otimes W\subset\cH_{A_1A_2:B_1B_2}$.	
\end{theorem}
\begin{proof}
	Let the two subspaces $V=\lin\{\ket{\a_1},...,\ket{\a_n}\}$ and $W=\lin\{\ket{\b_1},...,\ket{\b_n}\}$. Let $\ket{\Ps}\in V\otimes W \otimes \cH_{ab}$ be an arbitrary bipartite state of system $A_1A_2a$ and $B_1B_2b$. Since $V$ satisfies (\ref{EOFth}), by regarding $W\otimes \cH_{ab}$ as a subspace in the $\cH_{ab}$ of (\ref{EOFth}), we may assume that
	\begin{eqnarray}
	\text{Tr}_{B_1}
	\proj{\Ps}=\sum_l q_l
	\proj{\m_l}_{A_1}\otimes	
	\proj{\n_l}_{A_2a:B_2b}.
	\end{eqnarray}
	Similarly, since $W$ satisfies (\ref{EOFth}), we may assume that
	\begin{eqnarray}
	\text{Tr}_{B_2}
	\proj{\n_l}=\sum_m r_{l,m}
	\proj{\m_{l,m}}_{A_2}\otimes	
	\proj{\n_{l,m}}_{ab}.
	\end{eqnarray}
	The above two equations imply that
	\begin{eqnarray}
	\begin{aligned}
	\text{Tr}_{B_1B_2}
	\proj{\Ps}=&\sum_{l,m}
	q_lr_{l,m}
	\proj{\m_l}_{A_1}\otimes	
	\proj{\m_{l,m}}_{A_2}\\&\otimes	
	\proj{\n_{l,m}}_{ab}.
	\end{aligned}
	\end{eqnarray}
	By respectively regarding $A_1A_2$ and $B_1B_2$ as the system $A$ and $B$ in (\ref{EOFth}), we have proven that tracing out system $B$ destroys the entanglement between $AB$ and $ab$. So the assertion holds. This completes the proof.	
\end{proof}

Next we investigate the case of many EB spaces with two ways. The first way follows the tensor product in Theorem \ref{thm:votimesw}. For presenting the second way, we refer to $V\oplus_B W$ as the $B$-direct sum of two subspaces $V,W\subseteq\cH_{AB}$ in terms of system $B$, i.e., if $\a\in\cB(V)$ and $\b\in \cB(W)$ then $\a_B\perp\b_B$. They have been used for the separability problem of multipartite PPT states of rank at most four \cite{Lin2013Separability}, and the characterization of $2\times d$ PPT states recently \cite{cd12}.

\begin{corollary}
	\label{cr:votimesw}
	Suppose $V_j\subset\cH_{A_jB_j}$ with $j=1,2,...,n$ are EB subspaces. Then
	
	(i) so is the bipartite subspace $V_1\otimes...\otimes V_n \subset\cH_{A_1...A_n:B_1...B_n}$;
	
	(ii) so is the bipartite subspace $V_1\op_B ... \op_B V_n$.
\end{corollary}
\begin{proof}
	(i) We refer to Lemma \ref{thm:votimesw} and the above definition of $B$-direct sum of two subspaces. Since each $V_{j}$ is a bipartite subspace satisfying (\ref{EOFth}), we have
	\begin{eqnarray}
	\begin{aligned}
	&\text{Tr}_{B_1B_2,...,B_n}\proj{\Ps}\\=&\sum_{i_{1},...,i_{n}}
	a_{i_{1}}a_{i_{1},i_{2}},...,a_{i_{1},...,i_{n}}
	\proj{\m_{i_{1}}}_{A_1}\\&\otimes	
	\proj{\m_{i_{1},i_{2}}}_{A_2}\otimes...\otimes	
	\proj{\m_{i_{1},...,i_{n}}}_{A_n}\\&\otimes	
	\proj{\n_{i_{1},...,i_{n}}}_{ab}.
	\end{aligned}
	\end{eqnarray}
	(ii) For any $\ket{\Psi}\in (V_{1}\oplus_{B} ... \oplus_{B}V_{n})\otimes\mathcal{H}_{ab}$, it can be expressed as $\ket{\Psi}\in (V_{1}\otimes\mathcal{H}_{ab})\oplus_{B} ... \oplus_{B} (V_{n}\otimes\mathcal{H}_{ab})$. For each subsystem $k\in\{1,...,n\}$, we obtain	\begin{equation}	\text{Tr}_{B_{k}}(\ket{\Psi}_{V_{k}ab}\bra{\Psi})=\sum_{i}q_{ki}\ket{\mu_{ki}}_{A}\bra{\mu_{ki}}\otimes\ket{\nu_{i}}_{ab}\bra{\nu_{i}}.
	\end{equation}
	So we have
	\begin{equation}
	\text{Tr}_{B_{1},B_{2},...,B_{n}}(\ket{\Psi}_{Vab}\bra{\Psi})=(\sum_{k,i}q_{ki}\ket{\mu_{ki}}_{A}\bra{\mu_{ki}})\otimes\ket{\nu_{i}}_{ab}\bra{\nu_{i}}.
	\end{equation}
	the equation satisfies (\ref{EOFth}), that is, system $A$ is separate with system $ab$.
\end{proof}

Using the above results one can construct more subspaces satisfying (\ref{EOFth}) applying the known ones, say those in Theorem \ref{thm:mx3}. Since any subspace of an EB space is still an EB space, we can construct an EB space of any given dimension. On the other hand, we show that the converse of Theorem \ref{thm:votimesw} also holds.

\begin{lemma}
	\label{thm:votimesw2}
	Suppose $V\subset\cH_{A_1B_1}$ and $W\subset\cH_{A_2B_2}$ are two bipartite subspaces, and $V$ is not an EB space. Then neither is the bipartite subspace $V\otimes W\subset\cH_{A_1A_2:B_1B_2}$.	
\end{lemma}
\begin{proof}
	Since $V$ does not satisfy (\ref{EOFth}), there is a pure state $\ket{\Ps}_{A_1B_1:ab}$ such that $\tr_{B_1} \proj{\Ps}$ is a bipartite entangled state of system $A_1$ and $ab$.
	If $\ket{\a}\in W$, then the state $\tr_{B_1B_2} (\proj{\Ps}\otimes\proj{\a})$ is a bipartite entangled state of system $A_1A_2$ and $ab$.
\end{proof}

So far we have constructed EB spaces in $\bbC^m\otimes\bbC^n$, and it has dimension at most $n$ by Lemma \ref{le:ebdim}. Thus the state whose range is in the EB space has rank at most $n$. In the following we  construct an entangled $m\times n$ state $\r$ whose EOF is additive, and $\rank\r>n$. An example is the tensor product of $\r_{Aa:Bb}={1\over6}(I_3)_A\otimes (I_2)_B\otimes\proj{\ps}_{ab}$ where $\ket{\ps}$ is the two-qubit Bell state. So $\r_{Aa:Bb}\in\cB(\bbC^6\otimes\bbC^4)$, and has rank six. It follows from \cite{Chen2007Entanglement} that $\r_{Aa:Bb}$ has additive EOF and its EOF is equal to $1$ ebit.  This example shows an idea of constructing states with additive EOF different from the means by EB spaces. In particular, it requires only that the two states have additive EOF. Nevertheless, constructing states with additive EOF itself is a hard problem, even for two-qubit states. In contrast, we have characterized EB spaces in the space $\bbC^m\otimes\bbC^3$ and
$\bbC^2\otimes\bbC^n$, which have much larger dimensions.

\section{More applications}
\label{sec:moreapp}

In this section we apply our result of previous sections to the generalized model of two atoms, $A$ and $B$. They occupy the same site and collectively interacting with the quantized single mode of the field $C$ (in the absence of any radiation damping) known as the Tavis-Cummings  model \cite{PhysRev.170.379, PhysRev.188.692}. In the following we study a special case of the model in which the entanglement cost is equal to the EOF. 

Let us consider the model in the space $\bbC^2\otimes \bbC^2\otimes \bbC^3$. At the time $t$ the EOF is calculated, and we refer readers to \cite{PhysRevA.85.022320} for a detailed explanation. We show that the TC model belongs to EB space as well. The initial state is \begin{eqnarray}
\ket{\psi(0)}=(\alpha\ket{0,0}+\beta\ket{1,1})_{AB}\ket{0}_{C}.\end{eqnarray} The cavity mode will evolve within a Hilbert space spanned by $\{\ket{0},\ket{1},\ket{2}\}$, and the atomic system will evolve within the subspace $\{\ket{0,0},\ket{1,1},(\ket{0,1}+\ket{1,0})/\sqrt2\}$, the system at time $t$ is described by the state
\begin{eqnarray}
\begin{aligned}
\ket{\psi(t)}=&c_{1}(t)\ket{0,0}_{AB}\ket{2}_{C}+\frac{1}{\sqrt{2}}c_{2}(t)(\ket{0,1}+\ket{1,0})_{AB}\ket{1}_{C}\\&+c_{3}(t)\ket{1,1}_{AB}\ket{0}_{C}+c_{4}(t)\ket{0,0}_{AB}\ket{0}_{C},
\end{aligned}
\end{eqnarray}
where \begin{eqnarray*}
\begin{aligned}
c_{1}(t)&=-\frac{\sqrt{2}}{3}\beta(1-\cos(\sqrt{6}gt)),\\
c_{2}(t)&=-\frac{i}{\sqrt{3}}\beta\sin(\sqrt{6}gt),\\
c_{3}(t)&=\beta(1-\frac{1}{3}(1-\cos(\sqrt{6}gt))),\\
c_{4}(t)&=\alpha,
\end{aligned}
\end{eqnarray*}
and $\a$ and $\b$ are real.

One can show that the range of bipartite reduced density operator $
\r_{AC}=\tr_B\proj{\ps(t)}$ is spanned by $\{\ket{0,0}+\ket{1,1},\quad \ket{0,1}+\ket{1}(d\ket{0}+fe^{i\t}\ket{1}+g\ket{2})\}$, where the coefficients $ d=\frac{\sqrt{2}c_{4}(t)}{c_{2}(t)}, fe^{i\theta}=\frac{c_{2}(t)-\sqrt{2}c_{3}(t)}{c_{2}(t)}$ and $g=\frac{\sqrt{2}c_{1}(t)i}{c_{2}(t)}$. So it is an EB space when $d,f,g$ satisfy \eqref{eq:mx3,family3} in Theorem \ref{thm:2dim=3x2}. Since the EOF of $\r_{AC}$ is calculated in \cite{PhysRevA.85.022320}, we can obtain the  entanglement cost $E_{c}(\rho_{AC})=E_{f}(\rho_{AC})$.

Note that the entanglement cost involves the use of infinitely many copies of two atoms, say $A_1,..,A_n$ at system $A$, $B_1,...,B_n$ at system $B$, and $n\ra\infty$. Although they still interact with the fields $C_1,...,C_n$ at field $C$, our result 
shows that the use of many pairs of two atoms  $A_j,B_j$ in Tarvis-Cummings model does not decrease the entanglement between system $A$ and $C$. This is different from some other quantum-information tasks, such as the conversion from copies of three-qubit GHZ state to that of three-qubit W states. The latter relies on the fact that two copies of three-qubit W state has a tensor rank smaller than nine, namely the square of tensor rank of the W state \cite{Chen2018The}.

\section{Conclusions}
\label{sec:con}

We have studied the EB subspace for the additivity of EOF.
We have constructed the EB spaces in $\bbC^m\otimes\bbC^3$, and
applied our results to construct the EB space of any dimension. We have shown that the example in \cite{PhysRevLett.89.027901} is a special case of our results. We further have worked out the entanglement cost of a qubit-qutrit state in terms of the two-atom system of the Tavis-Cummings model. An open problem is to extend our results to the $2$-dimensional $m\times n$ EB subspace, so that more $m\times n$ states of rank two with additive EOF can be constructed. Another problem is to determine whether a given bipartite space is an EB space.

\section*{Acknowledgements}
This work was supported by the NNSF of China (Grant No. 11871089), and the Fundamental Research Funds for the Central Universities (Grant Nos. KG12040501, ZG216S1810 and ZG226S18C1).

\appendix

\section{The proof of Theorem \ref{thm:2dim=3x2} }
\label{app:2dim=3x2}

\begin{proof}
Suppose $V$ is a $2$-dimension subspace in $\cH_{AB}=\bbC^m\otimes\bbC^3$. One can verify the ``if" part of assertion using Lemma \ref{le:rank2} (i). In particular, the state $\r_{A:ab}$ from \eqref{eq:mx3,family3} is a two-qubit state. So it is separable if and only if it is PPT.

In the following we prove the ``only if" part of assertion. Assume that $V$ is a $2$-dimensional EB space. Our proof consists of two cases, namely (i) $V$ has a product state, (ii) $V$ has no product state.

(i) Up to LU equivalence we may assume that the space $V$ is spanned by the basis $\{\ket{\alpha_{1}, \alpha_{2}}, \ket{a,b}+\ket{c,d}+\ket{e,f}\}$.  Using the LU equivalence, we can rewrite the basis as
\begin{eqnarray}
	\begin{aligned}
	U\ket{\alpha_{1}, \alpha_{2}}=&\ket{0, 0},\\
	U(\ket{a,b}+\ket{c,d}+\ket{e,f})=&\ket{s,0}+\ket{t,1}+\ket{u,2}.
	\end{aligned}
\end{eqnarray}
If $\ket{s}$ is proportional to $\ket{0}$, then $V$ is spanned by $\ket{0,0}$ and $\ket{t,1}+\ket{u,2}$. By Definition \ref{df:ebc} the space spanned by $
	\{\ket{0,0}, \ket{1,1}+\ket{2,2}\}$ is
	EB-convertible to $V$. The ``only if" part of assertion holds.
	
On the other hand, suppose that $\ket{s}$ is not
proportional to $\ket{0}$. We consider the
bipartite state
	\begin{eqnarray}
	\begin{aligned}
	\ket{\Psi}_{Vab}=&\ket{0,0}_{AB}\ket{0}_{ab}+(\ket{s,0}+\ket{t,1}+\ket{u,2})_{AB}\ket{1}_{ab},
	\\
	\r_{A:ab}=&{\tr}_{B}\proj{\Psi}_{Vab}.
	\end{aligned}
\end{eqnarray}
Then the state $\r_{A:ab}$ is entangled in terms of Lemma \ref{le:rank2} (iii).
It is a contradiction with the assumption that $V$ is an EB space. We have proven the ``only if" part of assertion for case (i).

(ii) The 2-dimensional EB subspace $V\subset\cH_{AB}=\bbC^m\otimes\bbC^3$ has no product state. To prove the ``only if" part of assertion,
it suffices to consider the bipartite entangled states $\ket{\Psi_1}_{Vab}\in V\otimes\cH_{ab}$, and the bipartite reduced density operator
\begin{eqnarray}
\label{eq:sAab}
\s_{A:ab}={\tr}_{B}\proj{\Psi_1}_{Vab}.
\end{eqnarray}
Since  $V\subset\bbC^m\otimes\bbC^3$, we have $\rank\r_B\le3$.

(ii.a) If $\rank\r_B<3$ then $V$ is a $2$-dimensional EB subspace of $\bbC^m\otimes\bbC^2$ up to LU equivalence. It follows from Lemma \ref{le:2dim=2qubit} that $V$ is spanned by
	$
	\{\ket{x,0}, \ket{y,1}\}$ up to LU equivalence. So the space spanned by $
	\{\ket{0,0}, \ket{1,1}+\ket{2,2}\}$ is EB convertible to $V$. We have proven the ``only if" part of assertion for case (ii.a).

On the other hand, let $\rank\r_B=3$.
Eq. \eqref{eq:sAab} implies that
\begin{eqnarray}
\label{eq:rankaab}	
\rank\s_{A:ab}=\rank\r_B=3.
\end{eqnarray}
Since $V$ is an EB space, the bipartite state $\s_{A:ab}$ is separable. Lemma \ref{le:rank2} (ii) implies that $\rank\s_A=2$ or $3$. Up to LU equivalence we can assume that the range of $\s_A$ is a subspace of $\bbC^3$, and $\rank\s_A=2$ or $3$.

(ii.b) Let $\rank\s_A=3$. Using \eqref{eq:rankaab} and Lemma \ref{le:rank2} (ii), we obtain that $\ket{\Ps_1}_{Vab}$ has tensor rank three. That is, $\ket{\Ps_1}_{Vab}$ is the linear combination of exactly three tripartite product vectors in the space $\cH_{A}\otimes\cH_B\otimes\cH_{ab}
=\bbC^3\otimes\bbC^3\otimes\bbC^2$.
So up to SLOCC equivalence, $\ket{\Ps_1}_{Vab}$ is one of the first two states in \eqref{eq:3x3x2} using Lemma \ref{le:3x3x2}. Further, the assumption at the beginning of proof of (ii) says that $V\subset\cH_{AB}$ has no product state. So up to SLOCC equivalence $\ket{\Ps_1}_{Vab}$ is the second state in \eqref{eq:3x3x2}, i.e.,
\begin{eqnarray}\begin{aligned}
\ket{\Ps_1}=&
(W\otimes X\otimes Y)\\
&(\ket{0,0,0}+\ket{1,1,1}+\ket{2,2,0}+\ket{2,2,1}),	
\end{aligned}
\end{eqnarray}
for invertible matrices $W,X$ and $Y$.
So the space $V$ is spanned by $(W\otimes X)(\ket{0,0}+\ket{1,1})$  and $(W\otimes X)(\ket{1,1}+\ket{2,2})$. Hence the space spanned by $\{\ket{0,0}+\ket{1,1},\ket{1,1}+\ket{2,2}\}$ is EB convertible to $V$. We have proven the "only if" part of assertion for case (ii.b).

(ii.c) Let $\rank\s_A=2$. Recall the assumption at the beginning of proof of (ii), namely the space $V$ has no product vector.
Table 1 of \cite{PhysRevA.74.052331} shows that $V$ is spanned by $(W\otimes X)(\ket{0,0}+\ket{1,1})$  and $(W\otimes X)(\ket{0,1}+\ket{1,2})$ with invertible matrices $W,X$. In the following we show that $V$ is LU equivalent to the space in \eqref{eq:mx3,family3}. It will prove the "only if" part of assertion for case (ii.c).

Up to EB equivalence we may assume that $X=[x_{ij}]=[\ket{x_0},\ket{x_1},\ket{x_2}]$ is upper triangular, $x_{jj}>0$, and
\begin{eqnarray}
\label{eq:psi1}
\ket{\Psi_1}_{Vab}
=&
(I_2\otimes X)(\ket{0,0}+\ket{1,1})\ket{0}
\notag\\+&
(I_2\otimes X)(\ket{0,1}+\ket{1,2})\ket{1}.
\end{eqnarray}
If we choose
$P=
\bma
1&0\\
-{x_{12}\over x_{11}}&1\\
\ema^{\otimes2}$ and $P'=
\bma
1\over \sqrt{x_{11}}&0\\
0& \sqrt{x_{11}}\over x_{22}\\
\ema^{\otimes2}$ on $\cH_A\otimes\cH_{ab}$,
then
\begin{eqnarray}
\label{eq:phi}
\begin{aligned}
\ket{\Ph}_{B:A:ab}
=&
(I_B\otimes	P'P)
\ket{\Psi_1}_{B:A:ab}
\\=&
\ket{0}\ket{0,0}
+
\ket{1}
(\ket{0,1}+\ket{1,0})
\\+&
(d\ket{0}+fe^{i\t}\ket{1}
+g\ket{2})\ket{1,1},
\end{aligned}
\end{eqnarray}
where the four parameters $g>0$ and $d,f,\t\ge0$. Note that the separability of $\s_{A:ab}$ in \eqref{eq:sAab} is invariant under SLOCC equivalence. So $\a_{A:ab}=\tr_B\proj{\Phi}_{B:A:ab}$ is a bipartite separable state. Since $\rank\s_A=2$ and $V$ is a $2$-dimensional EB subspace, we obtain that $\a_{A:ab}$ is a two-qubit separable state. It is equivalent to the condition
$\det\a_{A:ab}^\G\ge0$ \cite{PhysRevA.74.010302}.
Using \eqref{eq:phi}, this condition is equivalent to the inequality in \eqref{eq:dfgt}. So $V$ is LU equivalent to the space in \eqref{eq:mx3,family3}. We have proven the ``only if" part of assertion for case (ii.c).
\end{proof}

The following fact is from Table 1 of \cite{PhysRevA.74.052331}.
\begin{lemma}
\label{le:3x3x2}
Consider the set of tripartite pure states $\r_{ABC}$, such that $\rank\r_A=\rank\r_B=3$ and $\rank\r_C=2$. The set has exactly six SLOCC-inequivalent states as follows.
\begin{eqnarray}
\label{eq:3x3x2}
\begin{aligned}
&\ket{0,0,0}+\ket{1,1,1}+\ket{2,2,0},\\
&\ket{0,0,0}+\ket{1,1,1}+\ket{2,2,0}+\ket{2,2,1},\\
&\ket{1,0,0}+\ket{0,1,0}+\ket{1,2,1}+\ket{2,1,1},\\
&\ket{0,0,1}+\ket{1,0,0}+\ket{0,1,0}+\ket{1,2,1}+\ket{2,1,1},\\
&\ket{0,0,1}+\ket{1,0,0}+\ket{0,1,0}+\ket{2,2,0},\\
&\ket{0,0,1}+\ket{1,0,0}+\ket{0,1,0}+\ket{2,2,1}.
\end{aligned}
\end{eqnarray}	
\qed
\end{lemma}
One can easily show that the first two states in \eqref{eq:3x3x2} has tensor rank  three, and the last four states has tensor rank greater than three. For the latest progress on tensor rank, we refer readers to the paper \cite{Chen2018The}.

\bibliography{entanglementcost}
\bibliographystyle{unsrt}

\end{document}